\documentclass[10pt, conference]{IEEEtran}

\usepackage{epsfig}
\usepackage{subfigure}
\usepackage{graphicx}

\usepackage{amsmath}
\usepackage{amssymb}
\usepackage{color}
\usepackage{cite}

\def\BibTeX{{\rm B\kern-.05em{\sc i\kern-.025em b}\kern-.08em
    T\kern-.1667em\lower.7ex\hbox{E}\kern-.125emX}}

\newtheorem{theorem}{Theorem}

\newtheorem{remark}{Remark}

\title{Information Theory of Matrix Completion}

\author{
\authorblockN{Changho Suh}
\authorblockA{KAIST, Daejeon, South Korea \\
Email: $\mathsf{chsuh@kaist.ac.kr}$}
}

\begin{document}

\IEEEoverridecommandlockouts

\maketitle

\begin{abstract}
\emph{Matrix completion} is a fundamental problem that comes up in a variety of applications like the Netflix problem, collaborative filtering, computer vision, and crowdsourcing. The goal of the problem is to recover a $k$-by-$n$ unknown matrix from a subset of its noiseless (or noisy) entries. We define an information-theoretic notion of \emph{completion capacity} $C$ that quantifies the maximum number of entries that one observation of an entry can resolve. This number provides the minimum number $m$ of entries required for reliable reconstruction: $m=\frac{kn}{C}$. Translating the problem into a distributed joint source-channel coding problem with encoder restriction, we characterize the completion capacity for a wide class of \emph{stochastic} models of the unknown matrix and the observation process. Our achievability proof is inspired by that of the Slepian-Wolf theorem. For an arbitrary stochastic matrix, we derive an upper bound on the completion capacity.


\end{abstract}

\section{Introduction}


Matrix completion has received considerable attention with applications in the Netflix problem~\cite{netflix}, collaborative filtering~\cite{CollaborativeFilter}, computer vision~\cite{ComputerVision}, and crowdsourcing~\cite{Oh:OR14}.
The objective of the problem is to reconstruct a $k$-by-$n$ unknown matrix from a given subset of noiselessly (or noisily) observed entries. Clearly completely irrelevant entries, if not observed, have no chance to be recovered, since other available entries cannot provide any information about the missing ones. In many applications, however, there can be correlation between entries. For instance, in the Netflix problem, movie ratings of a subscriber, comprising the elements of a single column or row, could play a role to infer ratings of other subscribers having similar tastes in movies. So a natural question that arises in this context is: what is the minimum number of entries needed for reconstruction of the unknown matrix with possibly correlated entries?

There has been a proliferation of reconstruction algorithms~\cite{CandesRecht,CandesTao:it,KMO:it,Recht} which may suggest an answer. Under the low-rank assumption of the matrix, Candes-Recht~\cite{CandesRecht} made use of the nuclear-norm-based optimization framework to provide a bound on the minimum number of entries required for reconstruction. This bound was further improved by Candes-Tao~\cite{CandesTao:it}, Keshavan-Montanari-Oh~\cite{KMO:it}, and Recht~\cite{Recht}. However, the best known bound thus far still comes with a numerical constant and/or logarithmic factor gap to the optimal performance, even under the low rank assumption;
hence the question raised above has not been answered completely. The challenge might come from the matrix model and the \emph{worst-case} analysis approach that these algorithms have pursued. The algorithms consider a \emph{deterministic} matrix and strive to guarantee reconstruction for \emph{all possible instances} of matrix entries. Moreover, the design of these algorithms is based primarily on computational considerations which might give challenge in finding the fundamental limits that can be achieved without any complexity restriction.

In an effort to make progress towards answering the question, we invoke Shannon's idea in~\cite{Shannon:IT}. In the paper, Shannon introduced a \emph{probabilistic} view to communication systems to characterize the maximum amount of information that can be reliably transmitted over a channel. Specifically Shannon viewed the information source and the channel as \emph{random} quantities that can be modeled via probability distributions. Given the stochastic source and channel, he then established the maximum rate of communication as the ratio of the channel capacity to the source entropy rate.

We advocate Shannon's approach for our problem. We model the unknown matrix and the observation process as \emph{stochastic} processes, and seek to find the maximum number of matrix entries that one noiseless (or noisy) observation can reveal on the average. We define the quantity as \emph{completion capacity} $C$. This number can provide an answer to the question that we raised before. The minimum number $m$ of entries needed for reliable reconstruction of a matrix is the ratio of the total number of matrix entries to the completion capacity: $m=\frac{kn}{C}$.

In this work, we characterize the completion capacity for a wide class of stochastic models of the unknown matrix and the observation process. Drawing parallels between the observation process and the communication channel, we translate the matrix completion problem into a distributed joint source-channel coding problem with encoder restriction.
We then invoke the achievability idea of the Slepian-Wolf theorem~\cite{SlepianWolf:it73}, established for a similar yet different setting, to develop a reconstruction algorithm. We show the optimality of this algorithm for
the case in which 
 a sequence of column vectors (or row vectors) of the unknown matrix is a \emph{stationary ergodic} process.
 This is the only assumption that we make on the statistics of the matrix: we consider \emph{arbitrary} distributions for the entries within each column (or row). The optimality is shown also for arbitrary observation processes which include the \emph{noisy} observation case as well.
For an arbitrary stochastic matrix, we develop an upper bound on the completion capacity.

\emph{Related Work:}
The denoising problem~\cite{Donoho95Denoising,Donoho95waveletshrinkage, UniveralDenoising:it05} bears similarity to our translated problem in that it intends to estimate an original signal from a noisy version.
Especially the problem setting in~\cite{UniveralDenoising:it05} is intimately related to ours since~\cite{UniveralDenoising:it05} takes a stochastic view on the original signal and the noisy channel as we do here.
However, the problem
does not focus on the setting in which only a few entries are observed.
Hence the performance metric of our interest, the minimum number of entries for reconstruction, was not investigated therein. On the other hand, Motahari-Bresler-Tse~\cite{MotahariBreslerTse:it} recently applied Shannon's approach that we take here to the DNA sequencing problem.
Shannon's idea together with this recent work gave inspiration to our work.

%



\section{Problem Formulation}
\label{sec:model}

\begin{figure}[t]
\begin{center}
{\epsfig{figure=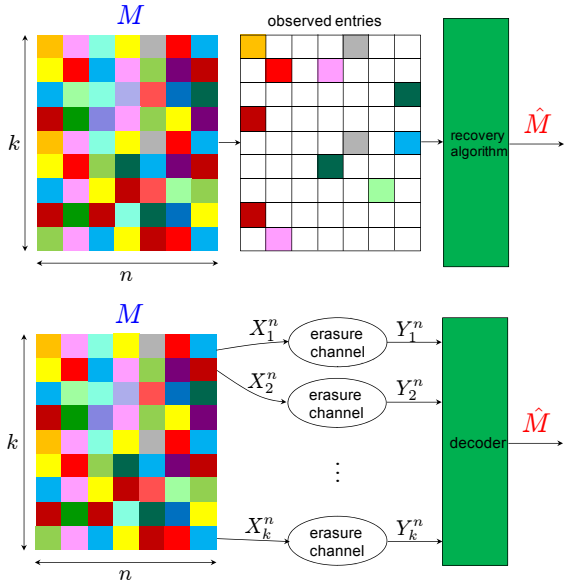, angle=0, width=0.41\textwidth}}
\end{center}
\caption{(Top): The matrix completion problem; (Bottom): A translated joint source-channel coding problem with  encoder restriction.}
\label{fig:framework}
\end{figure}

The problem of matrix completion is illustrated in the top of Fig.~\ref{fig:framework}.
Given partially observed entries of a $k$-by-$n$ unknown matrix $M$, the task of a recovery algorithm is to  decode $\hat{M}$. Here entries are assumed to be sampled uniformly at random. In this work, we translate the problem into a communication problem with encoder restriction. 
Specifically we model the observation process via an erasure channel with an erasure probability $1-p$, where  $p:=\frac{m}{kn}$ indicates \emph{observation rate} and $m$ denotes the number of observed entries. Note that the erasure channel is memoryless across entries due to the random sample assumption. Next we view $M$ as an information source that we wish to transmit. The problem can then be regarded as a point-to-point communication problem.

Here one noticeable distinction as compared to the conventional communication setting is that the information source $M$ cannot be processed (encoded) before transmission. In an effort to reflect this encoder restriction, we further reduce the problem into a \emph{distributed} joint source-channel coding problem in which joint processing across different transmitters is not allowed. In this setting, we are forced to somehow respect the encoder constraint, but not fully. Even in this setting, the encoder constraint is still applied to each transmitter.

See the bottom of Fig.~\ref{fig:framework} for the translated problem. Let $X_{\ell}^n$ be the $\ell$th row (or column) of the unknown matrix $M$ where $\ell \in [1:k]:=\{1,\cdots,k \}$. Without loss of generality, assume that $X_{\ell}^n$ denotes row components. Here we use shorthand notation to indicate the sequence up to $n$:  $X_{\ell}^n :=(X_{\ell 1}, \cdots, X_{ \ell n})$. The $(\ell,i)$ entry of the matrix, $X_{\ell i}$, is assumed to take values in a finite alphabet set. Since the observation process is modeled via an erasure channel, the channel output w.r.t. the $(\ell,i)$ entry is
\begin{align}
Y_{\ell i} = \left\{
               \begin{array}{ll}
                 X_{\ell i}, & \hbox{w.p. $p$;} \\
                 e, & \hbox{o.w.,}
               \end{array}
             \right.
\end{align}
where $\ell \in [1:k]$ and $i \in [1:n]$. Here the observation process is assumed to be noiseless. The noisy observation case will be dealt with in Section~\ref{sec:noisychannel}.
Decoder uses a function to estimate the unknown matrix $M$ from $\{ Y_{\ell}^n\}_{\ell=1}^{k}$.
An error occurs whenever $\hat{M} \neq M$. The average probability of error is given by
$\lambda = \mathbb{E} [ \textrm{Pr} (\hat{M} \neq M ) ]$.

We say that the \emph{completion rate} $R=\frac{kn}{m}= \frac{1}{p}$ is achievable if there exists a decoder function such that the average decoding error probability of $\lambda$ goes to zero as  $n$ tends to infinity. The \emph{completion capacity} $C$ is the supremum of the achievable completion rates. Note that this notion quantifies the maximum number of matrix entries that one non-erased observation can resolve on the average. With this notion $C$, the minimum number of entries required to reconstruct a $k$-by-$n$ matrix can be expressed as $\frac{kn}{C}$.

\begin{remark}
Since rows and columns are interchangeable  with a slight modification of notation, the completion capacity can be defined in the limit of $k$ alternatively. This asymptotic-regime analysis can be useful in practice as well since in many applications the matrix dimension admits a high dimensional setting, e.g., in the Netflix problem~\cite{netflix}, $(k,n) \approx ( 5 \cdot 10^5, 2 \cdot 10^4)$.
$\square$
\end{remark}

\begin{remark}
\label{remark:distributedsetting}
Our problem setting with multiple transmitters comes with another benefit. Since it matches the two-dimensional structure of the matrix, it helps to well capture the stochastic property of the matrix. In general, the statistics of the unknown matrix is governed by an arbitrary joint distribution. 
For some applications, however, the matrix may possess a column (or row) dependent stochastic property. For instance, consider a setting in the Netflix problem where the column index indicates movies of \emph{different} genres and the row index denotes subscribers with \emph{similar} tastes. The different genres may render movie ratings less correlated across columns. In the extreme case where the ratings are independent, we may have $p(x_{11},\cdots, x_{kn}) \approx \prod_{i=1}^n p(x_{1i},\cdots, x_{ki})$. On the other hand, ratings across rows may be strongly dependent as they are affected by similar tastes of subscribers. In an extreme case, $p(x_{\ell 1}, \cdots, x_{\ell n})$ may be the same for all $ \ell \in [1:k]$. The two-dimensional structure of our setting enables us to reflect this stochastic structure in a convenient manner. It turns out this makes analysis more tractable and simpler rather than in the point-to-point setting. This will be clearer in the subsequent sections. $\square$
\end{remark}

\section{I.I.D. Process Model}
\label{sec:iidmodel}

We first consider a simple stochastic model in which a sequence of column vectors $\{ (X_{1i}, \cdots, X_{ki}) \}_{i=1}^n$ is independent and identically distributed (i.i.d.) over $i$. Notice that this model represents the extreme case in Remark~\ref{remark:distributedsetting}, to some extent, where the movie scores are completely irrelevant across different genres. Here the entries within each column are assumed to follow an arbitrary distribution.
 Obviously this model is very simplistic and hence it may be far from the statistics of a realistic matrix in practical applications. But the i.i.d. model turns out to lay the foundation to tackle more complicated yet practically-relevant cases. Moreover it can serve as a lower bound to the capacity for the general case in which a sequence of column vectors are possibly dependent. This will be clearer in Sections~\ref{sec:ergodicmodel} and~\ref{sec:noisychannel}.

\begin{theorem}[I.I.D. Model]
\label{thm:iid_model}
\begin{align}
\label{eq:iidcapacity}
C = \frac{ \sum_{\ell=1}^{k} H(X_{\ell}) }{ H(X_1, \cdots, X_k) }.
\end{align}
\end{theorem}
\begin{proof}
See the following two subsections.
\end{proof}
\begin{remark}
For the case where the entries of a column vector are also i.i.d., i.e.,  $p(x_1,\cdots,x_k)=\prod_{\ell=1}p(x_{\ell})$, we get: $C =1$. This coincides with our intuition because in that case one observation cannot provide any information about others. For the other extreme case where $X_1=\cdots=X_k$, we have $C=k$. This result also makes sense since in that case one observation can reveal other $k-1$ entries in the same column. $\square$
\end{remark}

\subsection{Proof of Achievability}
\label{sec:iid_achievability}

Our proof bears a resemblance to that of the Slepian-Wolf theorem~\cite{SlepianWolf:it73}, developed in the context of a similar setting. The achievable scheme in~\cite{SlepianWolf:it73} employs random binning and \emph{joint typicality decoding}. Encoder $\ell$ independently partitions the space (that the sequence $X_{\ell}^n$ resides in) into multiple bins depending on a rate assigned. It then sends the bin index to which $X_{\ell}^n$ belongs. Upon receiving the bin indices from $k$ encoders, the decoder looks for a sequence tuple $(X_{1}^n, \cdots, X_{k}^n)$ such that their corresponding bin indices coincide with the received ones and the sequence tuple is jointly typical.

However, this scheme is not directly applicable to our case, since in our model 1) there is no explicit encoder; and 2) $X_{\ell}^n$'s are observed through erasure channels instead of through noiseless bit-pipes. These differences prevent employing the binning scheme. So the decoding rule needs to be modified accordingly. The key observation for the modification is that by the asymptotic equipartition property (AEP) of an i.i.d. process, $X_{\ell}^n$ is a typical sequence for sufficiently large $n$, and so is  $(X_{\ell}^n,Y_{\ell}^n)$~\cite{CoverThomas,ElGamalKim:NIT}. This leads us to consider the following decoding rule. For illustration purpose, we first focus on the case of $k=2$. The general case of an arbitrary value of $k$ will be dealt with later. Given $(y_1^n,y_2^n)$, the decoder tries to find a pair of sequences $(\hat{x}_1^n, \hat{x}_2^n)$ such that
  $(\hat{x}_1^n, y_1^n) \in A_{\epsilon}^{(n)} (X_1, Y_1)$, $(\hat{x}_2^n, y_2^n) \in A_{\epsilon}^{(n)} (X_2, Y_2)$ and $(\hat{x}_1^n, \hat{x}_2^n) \in A_{\epsilon}^{(n)} (X_1, X_2)$,
where $A_{\epsilon}^{(n)} (X_1, Y_1)$ denotes a joint typical set w.r.t. $X_1^n$ and $Y_1^n$ for some $\epsilon >0$.
Similarly $A_{\epsilon}^{(n)} (X_2, Y_2)$ and $A_{\epsilon}^{(n)} (X_1, X_2)$ are defined accordingly.
For notational simplicity, we will use a common notation $A_{\epsilon}^{(n)}$, which can be easily differentiated from contexts.

Now let us consider the probability of error. Using the union bound, we can upper-bound the probability of error as:
\begin{align*}
& \textrm{Pr} \left \{ (\hat{X}_1^n, \hat{X}_2^n) \neq (X_1^n, X_2^n) \right \} \leq \textrm{Pr} \left \{ (X_1^n, X_2^n) \notin A_{\epsilon}^{(n)} \right \} \\
&\quad + \sum_{(x_1^n,x_2^n) \in A_{\epsilon}^{(n)} } p(x_1^n,x_2^n) \left \{ \textrm{Pr} ( E_1) +   \textrm{Pr} ( E_2)+ \textrm{Pr} ( E_{12}) \right \}
\end{align*}
where $E_{\ell}$ indicates the event that $\hat{X}_{\ell}^n \neq x_{\ell}^n$, and $E_{12}$ denotes the event that $\hat{X}_1^n \neq x_1^n$ and $\hat{X}_2^n \neq x_2^n$. By AEP, $\textrm{Pr} \{ (X_1^n, X_2^n) \notin A_{\epsilon}^{(n)} \} \rightarrow 0$ as $n$ tends to infinity. Now consider
\begin{align*}
\textrm{Pr} (E_1) &= \sum_{ \hat{x}_1^n \neq x_1^n, \hat{x}_1^n \in A_{\epsilon}^{(n)} (X_1|x_2^n) } \textrm{Pr} \left \{ (\hat{x}_1^n, y_1^n) \in A_{\epsilon}^{(n)} \right \} \\
&\leq 2^{-n ( I(X_1; Y_1) -H(X_1|X_2) - 2 \epsilon) }
\end{align*}
where the inequality follows from $|A_{\epsilon}^{(n)}(X_1|x_2^n)| \leq 2^{n (H(X_1|X_2) + \epsilon)}$ for $x_2^n \in A_{\epsilon}^{(n)}$~\cite{CoverThomas,ElGamalKim:NIT}. Similarly we get:
$\textrm{Pr} (E_2) \leq 2^{-n (I(X_2; Y_2) - H(X_2|X_1) - 2 \epsilon) }$ and $\textrm{Pr} (E_{12}) \leq 2^{-n (I(X_1;Y_1) + I(X_2; Y_2) - H(X_1, X_2) - 3 \epsilon) }$.
Hence, the probability of error can be made arbitrarily close to zero (as $n$ tends to infinity) if
$I(X_1;Y_1) \geq H(X_1|X_2)$, $I(X_2;Y_2) \geq H(X_2|X_1)$, and $I(X_1;Y_1) + I(X_2;Y_2) \geq H(X_1,X_2)$.
Note that $I (X_{\ell}; Y_{\ell}) = H (X_{\ell}) - H(X_{\ell} | Y_{\ell}) =p H(X_{\ell})$. Applying this to the above condition, we get:
 $p \geq \max \left \{ \frac{H(X_1|X_2)}{H(X_1)}, \frac{H(X_2|X_1)}{H(X_2)}, \frac{H(X_1,X_2)}{H(X_1) + H(X_2)}  \right \}$.
Using the non-negativity of entropy and the fact that conditioning reduces entropy, we get $p \geq \frac{H(X_1,X_2)}{H(X_1) + H(X_2)}$. Since $R = \frac{1}{p}$, this condition becomes $R \leq \frac{H(X_1) + H(X_2)}{H(X_1, X_2)}$, which proves the achievablility for $k=2$.

Precisely similar arguments can be made for an arbitrary value of $k$. Specifically, given $(y_1^n, \cdots, y_k^n)$, the decoder looks for a tuple of $k$ sequences $(\hat{x}_1^n, \cdots, \hat{x}_k^n)$ such that $(\hat{x}_{\ell}^n, y_{\ell}^n) \in A_{\epsilon}^{(n)}, \forall \ell \in [1:k]$ and $(\hat{x}_{1}^n,\cdots, \hat{x}_k^n) \in A_{\epsilon}^{(n)}$. Then, the probability of error can be made arbitrarily close to zero if $\sum_{\ell \in {\cal S}} I (X_{\ell}; Y_{\ell}) \geq H(X_{\cal S}|X_{ {\cal S}^c} ), \forall {\cal S} \subseteq [1:k]$, where $X_{\cal S} = \{ X_{\ell} \}_{\ell \in {\cal S}}$. This yields the desired result:
 $R \leq \min_{{\cal S} \subseteq [1:k]} \frac{ \sum_{\ell \in {\cal S}} H(X_{\ell})  }{H(X_{\cal S}|X_{ {\cal S}^c} )} = \frac{ \sum_{\ell=1}^{k} H(X_{\ell})  }{ H(X_1,\cdots, X_{k}) }$.

\subsection{Proof of Converse}
\label{sec:iid_converse}

Using $I(X_{\ell}; Y_{\ell}) = p H(X_{\ell})$ and Fano's inequality
 $H( X_1^n, \cdots, X_{k}^n | Y_1^n, \cdots, Y_k^n) \leq n \epsilon_n$,
we get:
\begin{align*}
&\frac{ n \sum_{\ell =1}^{k} H(X_{\ell}) }{ R } = \sum n I(X_{\ell}; Y_{\ell}) \overset{(a)}{=} \sum I( X_{\ell}^n; Y_{\ell}^n ) \\
& \quad \overset{(b)}{\geq} H \left( Y_{1}^n, \cdots, Y_k^n \right) - \sum H( Y_{\ell}^n | X_{\ell}^n) \\
& \quad \overset{(c)}{=}  I \left( X_{1}^n, \cdots, X_k^n ; Y_{1}^n, \cdots, Y_k^n  \right) \\
& \quad \overset{(d)}{\geq}  n H(X_1, \cdots, X_k ) - n \epsilon_n
\end{align*}
where $(a)$ follows from the memoryless property of the channel and the i.i.d. assumption on $X_{\ell}^n$; $(b)$ follows from the fact that conditioning reduces entropy; $(c)$ follows from a Markov chain of $ \{ X_{j}^n, Y_{j}^n \}_{j \neq \ell} \rightarrow  X_{\ell}^n \rightarrow Y_{\ell}^n$; and $(d)$ follows from Fano's inequality. If $R$ is achievable, then $\epsilon_n \rightarrow 0$ as $n$ tends to infinity. Hence, we obtain the desired bound, as claimed.

\section{Stationary Ergodic Process Model}
\label{sec:ergodicmodel}

We extend the result for the i.i.d. case to a more general case in which a sequence of column vectors $\{(X_{1i},\cdots, X_{ki})\}_{i=1}^{\infty}$ is now a stationary ergodic process. As before, the entries within each  column are assumed to follow an arbitrary distribution.

\begin{theorem}[Stationary Ergodic Model]
\label{thm:ergodic_model}
\begin{align}
\label{eq:ergodiccapacity}
C = \frac{ \sum_{\ell=1}^k   \bar{H}( X_{\ell} ) -  \sum_{\ell=1}^k  a_{\ell} }{ \bar{H}( X_1, \cdots, X_k) - \sum_{\ell =1}^k a_{\ell} }
\end{align}
where
$\bar{H}(X_{\ell}):=\lim_{n \rightarrow \infty} \frac{1}{n} H (X_{\ell}^n)$ indicates the entropy rate, $\bar{H}(X_1, \cdots, X_k) :=$ $\lim_{n \rightarrow \infty}$ $\frac{1}{n} H (X_1^n, \cdots, X_k^n)$,
 and
$a_{\ell} :=$ $\lim_{n \rightarrow \infty}$ $\frac{1}{n} \sum_{i=1}^n I ( Y_{\ell, (i+1)}^{n}; X_{\ell i} | X_{\ell}^{i-1})$.
\end{theorem}
\begin{proof}
See the following two subsections.
\end{proof}
\begin{remark}
Unlike the i.i.d. case, here we have $a_{\ell}$. One can interpret this quantity as the information that the present provides about the future noisy signals knowing the past. So this can be alternatively expressed via directed information~\cite{Massey:ITA}. In the stationary ergodic case, correlation between entries can make $a_{\ell}$ strictly positive, and this contributes to increasing completion capacity. This result makes sense since with more correlation, obviously more entries can be revealed per observation. Also note that this theorem includes the result for the i.i.d. model as a special case. Specializing to the i.i.d. model, we get: $a_{\ell}=0$, $\bar{H}( X_{\ell}) = H(X_{\ell})$, $\forall \ell$ and $\bar{H}(X_1,\cdots, X_{k}) = H(X_1,\cdots, X_k)$, thus yielding~\eqref{eq:iidcapacity}. In general, $a_{\ell} \geq 0$, $\bar{H}(X_{\ell}) \leq H( X_{\ell})$, and $\bar{H}(X_{1}, \cdots, X_{k}) \leq H( X_{1},\cdots, X_{k})$. From this, we see that the capacity in the i.i.d. case serves as a lower bound, as claimed earlier. $\square$
\end{remark}


\subsection{Proof of Achievability}
\label{sec:ergodic_achievability}
The proof relies primarily on the Shannon-McMillan-Breiman theorem~\cite{Shannon:IT,McMillan:stat,Breiman:stat}, which showed that for a stationary ergodic process, the AEP holds as in the i.i.d. case:
$- \frac{1}{n} \log p(X_1^n, \cdots, X_k^n) \rightarrow \bar{H}(  X_1, \cdots, X_k )$ and $- \frac{1}{n} \log p(X_{\ell}^n) \rightarrow \bar{H}( X_{\ell} )$
with probability 1. Using the stationarity of $\{X_{\ell i}\}_{i=1}^{\infty}$ and the memoryless property of the channel, one can show that $\{ X_{\ell i}, Y_{\ell i} \}_{i=1}^{\infty}$ is also stationary and ergodic $\forall \ell \in [1:k]$. Moreover, $\{Y_{\ell i} \}_{i=1}^{\infty}$ is necessarily stationary and ergodic since it is a projection of $\{ X_{\ell i}, Y_{\ell i} \}_{i=1}^{\infty}$. Hence, the AEP holds for these random processes as well.
This naturally motivates us to apply the same decoding rule that we used in the i.i.d. case. Given $(y_1^n, \cdots, y_k^n)$, the decoder looks for a unique tuple of $k$ sequences $(\hat{x}_1^n, \cdots, \hat{x}_k^n)$ such that $(\hat{x}_{\ell}^n, y_{\ell}^n) \in \bar{A}_{\epsilon}^{(n)}, \ell \in [1:k]$ and $(\hat{x}_{1}^n,\cdots, \hat{x}_k^n) \in
\bar{A}_{\epsilon}^{(n)}$. Here $\bar{A}_{\epsilon}^{(n)}$ indicates a joint typical set w.r.t. some stationary ergodic processes.
One can then readily show that the probability of error goes to zero if
$\sum_{\ell \in {\cal S}} \bar{I}( X_{\ell}; Y_{\ell}) \geq \bar{H}(  X_{\cal S}| X_{{\cal S}^c}), \forall {\cal S} \subseteq [1:k]$,
where $\bar{I}(X_{\ell}; Y_{\ell}):=\lim_{n \rightarrow \infty} \frac{1}{n} I (X_{\ell}^n; Y_{\ell}^n)$ denotes the mutual information rate.

Unlike the i.i.d. case, in the stationary ergodic case, $\bar{I}(X_{\ell}; Y_{\ell}) \neq p \bar{H} (X_{\ell})$ in general. Actually we find the relation between $p = \frac{1}{R}$ and $\bar{I}( X_{\ell}; Y_{\ell})$ from the following computation:
\begin{align}
\begin{split}
\label{eq:IRcomp_relation}
I &(X_{\ell}^n; Y_{\ell}^n) = H(X_{\ell}^n) - \sum H(X_{\ell i} | X_{\ell}^{i-1}, Y_{\ell}^n ) \\
& = H(X_{\ell}^n) - \sum H(X_{\ell i} | X_{\ell}^{i-1}, Y_{\ell i}, Y_{\ell, (i+1)}^n) \\
&=  H(X_{\ell}^n) - (1-p) \sum H(X_{\ell i} | X_{\ell}^{i-1}, Y_{\ell, (i+1)}^n) \\
&=  p H(X_{\ell}^n)  + (1-p) \sum I (X_{\ell i}; Y_{\ell, (i+1)}^n | X_{\ell}^{i-1}) \\
\end{split}
\end{align}
where the second equality is due to a Markov chain of $Y_{\ell}^{i-1} \rightarrow X_{\ell}^{i-1} \rightarrow X_{\ell i}$. Notice that $a_{\ell}:= \lim_{n \rightarrow \infty} \frac{1}{n} \sum I (X_{\ell i}; Y_{\ell, (i+1)}^n | X_{\ell}^{i-1})$ \emph{exists} due to the stationarity of $(X_{\ell}^n, Y_{\ell}^n)$ and $X_{\ell}^n$. Using the above, we then get:
 $R \leq \min_{{\cal S} \subseteq [1:k]} \frac{ \sum_{\ell \in {\cal S}} \left \{ \bar{H}(X_{\ell}) - a_{\ell} \right \}  }{ \bar{H}( X_{\cal S}| X_{ {\cal S}^c} ) - \sum_{\ell \in {\cal S}} a_{\ell} } $.
 Using $\bar{H}(X_{\ell})- a_{\ell} \geq 0$ and the fact that conditioning reduces entropy, we finally obtain the desired result.

\subsection{Proof of Converse}
\label{sec:ergodic_converse}

Using~\eqref{eq:IRcomp_relation} and making similar arguments as in the i.i.d. case, we get:
\begin{align*}
&\frac{ \sum_{\ell =1}^{k} \{ \bar{H}(X_{\ell}) - a_{\ell} \} }{ R } + \sum a_{\ell} = \sum \bar{I} ( X_{\ell}; Y_{\ell}) \\
& \quad \geq \lim_{n \rightarrow \infty } \frac{1}{n} \left[ H \left( Y_{1}^n, \cdots, Y_k^n \right) - \sum H( Y_{\ell}^n | X_{\ell}^n) \right ]\\
& \quad =  \lim_{n \rightarrow \infty} \frac{1}{n} I \left( X_{1}^n, \cdots, X_k^n ; Y_{1}^n, \cdots, Y_k^n  \right) \\
& \quad \geq  \bar{H} ( X_1, \cdots, X_k ) - \lim_{n \rightarrow \infty} \epsilon_n.
\end{align*}
This yields the desired bound.


\section{Generalization}
\label{sec:noisychannel}

We extend the noiseless observation case to the noisy case illustrated in Fig.~\ref{fig:noisychannel}. We model the noisy observation via a discrete memoryless channel described by a conditional probability distribution. The noisy version $Y_{\ell}^{n}$ then passes through an erasure channel to produce $Z_{\ell}^n$, $\forall \ell$. We will first characterize the completion capacities for the i.i.d. model and for the stationary ergodic model. We will then derive an upper bound for an arbitrary stochastic matrix.


\begin{figure}[t]
\begin{center}
{\epsfig{figure=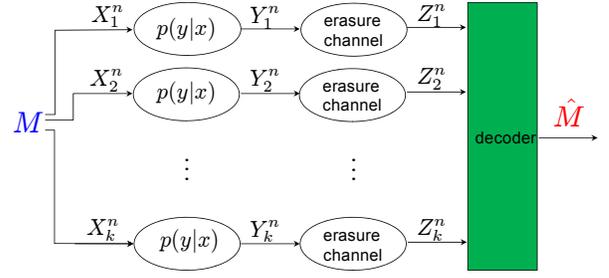, angle=0, width=0.43\textwidth}}
\end{center}
\caption{Noisy observation process.}
\label{fig:noisychannel}
\end{figure}

\begin{theorem}[Noisy Entries]
\label{thm:noisy_model}
For the i.i.d. process model,
\begin{align}
C = \frac{ \sum_{\ell=1}^{k} I(X_{\ell}; Y_{\ell}) }{ H(X_1, \cdots, X_k) }.
\end{align}
For the stationary ergodic process model,
\begin{align}
C = \frac{ \sum_{\ell=1}^{k} \bar{I}(X_{\ell}; Y_{\ell}) - \sum_{\ell=1}^{k} b_{\ell} }{ \bar{H}(X_1, \cdots, X_k) - \sum_{\ell=1}^{k} b_{\ell}}
\end{align}
where
$b_{\ell } := \lim_{n \rightarrow \infty} \frac{1}{n} \sum_{i=1}^n I(Y_{\ell i}; Z_{\ell, (i+1)}^n | Y_{ \ell }^{i-1})$.
\end{theorem}
\begin{remark}
Observe that the entropy and the entropy rate that appeared in the noiseless case (\eqref{eq:iidcapacity} and~\eqref{eq:ergodiccapacity}) are now replaced by the mutual information and the mutual information rate respectively. This clearly shows the reduction in capacity due to noises. $\square$
\end{remark}
\begin{proof}
The proof for the i.i.d. case is straightforward. The key observation here is that $(X_{\ell}^n,Z_{\ell}^n)$ is also i.i.d. This leads us to apply the same decoding rule as before. One can then easily obtain the following condition for achievability:
$\sum_{\ell \in {\cal S}} I( X_{\ell}; Z_{\ell}) \geq H(  X_{\cal S}| X_{{\cal S}^c}), \forall {\cal S} \subseteq [1:k]$.
We now consider:
$I (X_{\ell}; Z_{\ell}) = H(X_{\ell}) - (1-p) H(X_{\ell}) - p H(X_{\ell}|Y_{\ell}) = p I(X_{\ell}; Y_{\ell})$.
Applying this to the above condition, we can readily prove the achievability.
In the i.i.d. case, $I(X_{\ell}^n; Z_{\ell}^n)=p I(X_{\ell}^n; Y_{\ell}^n)$. So we get: $\frac{ \sum_{\ell =1}^{k} I(X_{\ell}^n; Y_{\ell}^n) }{ R}  = \sum I ( X_{\ell}^n; Z_{\ell}^n) \geq  nH ( X_1, \cdots, X_k ) - n\epsilon_n$. This establishes the converse proof.

Unlike the i.i.d. case, for the stationary ergodic model, $I(X_{\ell}^n; Z_{\ell}^n) \neq p I(X_{\ell}^n; Y_{\ell}^n)$ in general. Instead we get:
\begin{align}
\begin{split}
\label{eq:IXZ_computation}
&I (X_{\ell}^n; Z_{\ell}^n) = I(X_{\ell}^n; Y_{\ell}^n)  - I (X_{\ell}^n ; Y_{\ell}^n | Z_{\ell}^n) \\
&=  I(X_{\ell}^n; Y_{\ell}^n)  -\sum (1-p) I (X_{\ell}^n ; Y_{\ell i} | Y_{\ell}^{i-1}, Z_{\ell, (i+1)}^n) \\
&=  p I(X_{\ell}^n; Y_{\ell}^n)  + \sum (1-p)  \cdot \\
&\qquad \;\; \{ I (X_{\ell}^n ; Y_{\ell i} | Y_{\ell}^{i-1}) - I (X_{\ell}^n ; Y_{\ell i} | Y_{\ell}^{i-1}, Z_{\ell, (i+1)}^n) \} \\
&=  p I(X_{\ell}^n; Y_{\ell}^n) + \sum (1-p) I (Y_{\ell i}; Z_{\ell, (i+1)}^n | Y_{\ell}^{i-1})
\end{split}
\end{align}
where the first step follows from $X_{\ell}^n \rightarrow Y_{\ell}^n \rightarrow Z_{\ell}^n$;
 and the last step is due to $(Y_{\ell}^{i-1}, Z_{\ell, (i+1)}^n) \rightarrow  X_{\ell}^n \rightarrow Y_{\ell i}$. From~\eqref{eq:IXZ_computation}, we get:
$\bar{I}(X_{\ell}; Z_{\ell}) = p ( \bar{I} (X_{\ell}; Y_{\ell} ) - b_{\ell}) + b_{\ell}$.
Using this and making the same arguments as before, one can readily prove the achievability and the converse.
\end{proof}

\begin{theorem}[Arbitrary Stochastic Matrix]
\label{thm:upperbound}
\begin{align}
C \leq \limsup_{n \rightarrow \infty} \frac{ \sum_{\ell} I(X_{\ell}^n; Y_{\ell}^n) - \sum_{\ell} b_{\ell} (n) }{ H(X_1^n, \cdots, X_k^n) -  \sum_{\ell} (b_{\ell}(n) - c_{\ell} (n) )},
\end{align}
where
$b_{\ell}(n) := \sum_{i=1}^n I(Y_{\ell i}; Z_{\ell, (i+1)}^n | Y_{ \ell }^{i-1})$ and $c_{\ell}(n) := I(Z_{\ell}^n; Z_{1}^n, \cdots, Z_{\ell -1}^n)$.
\end{theorem}
\begin{proof}
Starting with~\eqref{eq:IXZ_computation}, we get:
\begin{align*}
&\frac{ \sum \{ I(X_{\ell}^n; Y_{\ell}^n) - b_{\ell} (n) \} }{ R} + \sum b_{\ell} (n) = \sum I ( X_{\ell}^n; Z_{\ell}^n) \\
& \quad = \left[ H \left( Z_{1}^n, \cdots, Z_k^n \right) - \sum H( Z_{\ell}^n | X_{\ell}^n) \right ] + \sum c_{\ell} (n) \\
& \quad =  I \left( X_{1}^n, \cdots, X_k^n ; Z_{1}^n, \cdots, Z_k^n  \right) + \sum c_{\ell} (n) \\
& \quad \geq  H ( X_1^n, \cdots, X_k^n ) + \sum c_{\ell} (n) - n \epsilon_n.
\end{align*}
This completes the proof.
\end{proof}

\section{Conclusion}
\label{sec:conclusion}



We established the completion capacity for a class of stochastic matrices, and developed an upper bound for arbitrary matrices. Our results can be used in evaluating different reconstruction algorithms, as well as provide the limits that an optimal algorithm can achieve. This paper comes with some limitations. The stationary ergodic model for the unknown matrix mainly considered herein may be far from the statistics of a realistic matrix in some important applications. Moreover, our reconstruction algorithm is based on joint typicality decoding which faces challenge in implementation. Hence future works of interest would be (1) characterizing the completion capacity for an arbitrary stochastic matrix; (2) developing an optimal low-complexity algorithm that achieves the limits possibly without knowing the statistics of the matrix.

%
%
%
%

\section*{Acknowledgment}
The author would like to thank Prof. Jinwoo Shin for discussions in the early stage of this work.

\bibliographystyle{ieeetr}
\bibliography{bib_MC}

\end{document}